\documentclass[proceedings,conferenceproceedings,submit,moreauthors,pdftex,10pt,a4paper]{mdpi} 

\preto{\abstractkeywords}{\nolinenumbers}

\usepackage{amssymb}
\usepackage{amsmath}
\usepackage[scr]{rsfso}
\usepackage{mathtools}
\usepackage{microtype}
\newtheorem{definition}[theorem]{Definition}
\newtheorem{lemma}[theorem]{Lemma}
\newtheorem{thm}[theorem]{Theorem}
\newtheorem{proposition}[theorem]{Proposition}

\renewcommand{\[}{\begin{equation}}
\renewcommand{\]}{\end{equation}}
\def\<#1\>{\begin{align}#1\end{align}}

\usepackage{csquotes}
\MakeOuterQuote{"}

\newcommand{\RegisterPairedDelimiter}[3][1]{
\ifnum#1=1 \newcommand{#2}[2][-1]{\ensuremath{\mathinner{%
\ifnum##1=-1 #3*{##2}\fi%
\ifnum##1=0 #3{##2}\fi%
\ifnum##1=1 #3[\big]{##2}\fi%
\ifnum##1=2 #3[\Big]{##2}\fi%
\ifnum##1=3 #3[\bigg]{##2}\fi%
\ifnum##1=4 #3[\Bigg]{##2}\fi%
}}}\fi%
\ifnum#1=2 \newcommand{#2}[3][-1]{\ensuremath{\mathinner{%
\ifnum##1=-1 #3*{##2}{##3}\fi%
\ifnum##1=0 #3{##2}{##3}\fi%
\ifnum##1=1 #3[\big]{##2}{##3}\fi%
\ifnum##1=2 #3[\Big]{##2}{##3}\fi%
\ifnum##1=3 #3[\bigg]{##2}{##3}\fi%
\ifnum##1=4 #3[\Bigg]{##2}{##3}\fi%
}}}\fi%
}

\DeclarePairedDelimiter\deldelim{(}{)}
\RegisterPairedDelimiter{\del}{\deldelim}

\DeclarePairedDelimiter\sbrdelim{[}{]}
\RegisterPairedDelimiter{\sbr}{\sbrdelim}

\DeclarePairedDelimiter\cbrdelim{\{}{\}}
\RegisterPairedDelimiter{\cbr}{\cbrdelim}

\DeclarePairedDelimiter\absdelim{|}{|}
\RegisterPairedDelimiter{\abs}{\absdelim}

\DeclarePairedDelimiter\normdelim{\lVert}{\rVert}
\RegisterPairedDelimiter{\norm}{\normdelim}

\DeclarePairedDelimiter\innerprodonedelim{\langle}{\rangle}
\RegisterPairedDelimiter{\innerprodone}{\innerprodonedelim}
\DeclarePairedDelimiterX\innerprodtwodelim[3]{\langle}{\rangle}{#1,#2}
\RegisterPairedDelimiter[2]{\innerprodtwo}{\innerprodtwodelim}
\newcommand{\innerprod}[2][-1]{\@ifnextchar\bgroup{\expandafter\innerprodtwo[#1]{#2}}{\innerprodone[#1]{#2}}}

\DeclarePairedDelimiterX\dualproddelim[3]{\langle}{\rangle}{#1\;\delimsize|\;\mathopen{}#2}
\RegisterPairedDelimiter[2]{\dualprod}{\dualproddelim}

\let\P\textP 
\DeclareMathOperator{\P}{\mathbb{P}} 
\DeclareMathOperator{\E}{\mathbb{E}} 
\newcommand{\I}[1][]{\operatorname{I}\ifx\\#1\\\else_{\{#1\}}\fi} 
\newcommand{\R}{\mathbb{R}} 
\newcommand{\iid}[1][]{\overset{\textrm{iid}}{\sim}\ifx\\#1\\\else\operatorname{#1}\fi} 
\newcommand{\given}{\mid} 
\let\textu\u 
\renewcommand{\u}{\relax\ifmmode\cup\else\expandafter\textu\fi} 
\let\texthat\^ 
\renewcommand{\^}{\relax\ifmmode\cap\else\expandafter\texthat\fi} 
\newcommand{\goesto}{\rightarrow} 
\let\textv\v 
\renewcommand{\v}{\relax\ifmmode\expandafter\boldsymbol\else\expandafter\textv\fi} 
\newcommand{\m}{\expandafter\mathbf} 
\newcommand{\f}{\expandafter\operatorname} 

\newcommand{\od}[3][]{\frac{\dif{^{#1}}#2}{\dif{#3^{#1}}}} 
\newcommand{\pd}[3][]{\frac{\partial{^{#1}}#2}{\partial{#3^{#1}}}} 
\newcommand\grad\nabla
\DeclareMathOperator{\dif}{d \!} 
\let\textd\d 
\renewcommand{\d}{\relax\ifmmode\dif\else\expandafter\textd\fi} 
\let\textc\c 
\renewcommand{\c}{\relax\ifmmode\expandafter\mathcal\else\expandafter\textc\fi} 

\newcommand{\distr}[1][]{\sim\ifx\\#1\\\else\operatorname{#1}\fi} 

\newcommand{\todo}[1]{\textcolor{red}{\ifmmode\text{TODO: #1}\else{TODO: #1}\fi}}

\let\textcite\citet
\newcommand{\printbibliography}{\bibliography{AVTcitations,local}}

\firstpage{1} 
\articlenumber{x}
\doinum{10.3390/------}
\pubvolume{xx}
\pubyear{2018}
\copyrightyear{2018}
\history{Received: date; Accepted: date; Published: date}

\Title{A Piecewise Deterministic Markov Process via $(r,\theta)$ swaps in hyperspherical coordinates}

\Author{Alexander Terenin $^{1}$, Daniel Thorngren $^{2}$}

\address{$^{1}$ \quad Department of Mathematics, Imperial College London, London, UK; a.terenin17@imperial.ac.uk\\
$^{2}$ \quad Department of Physics, University of California, Santa Cruz, CA, USA; dpthorngren@ucsc.edu}

\abstract{Recently, a class of stochastic processes known as piecewise deterministic Markov processes has been used to define continuous-time Markov chain Monte Carlo algorithms with a number of attractive properties, including compatibility with stochastic gradients like those typically found in optimization and variational inference, and high efficiency on certain big data problems.
Not many processes in this class that are capable of targeting arbitrary invariant distributions are currently known, and within one subclass all previously known processes utilize linear transition functions.
In this work, we derive a process whose transition function is nonlinear through solving its Fokker-Planck equation in hyperspherical coordinates.
We explore its behavior on Gaussian targets, as well as a Bayesian logistic regression model with synthetic data.
We discuss implications to both the theory of piecewise deterministic Markov processes, and to Bayesian statisticians as well as physicists seeking to use them for simulation-based computation.
}

\keyword{Bayesian Statistics, Big Data, Continuous-time MCMC, Markov Chain Monte Carlo, Piecewise Determinisitic Markov Process}

\begin{document}

\section{Introduction} \label{sec:intro}

The Bayesian statistical paradigm possesses many desirable properties, including the ability to quantify uncertainty about a set of estimated parameters.
However, using it entails the computation of posterior probability distributions -- this tends to be expensive, because these are inherently complicated, and depend on the data used to define them.
This is especially challenging in modern application areas such as natural language processing and analysis of internet data, which involve large data sets.
Creating algorithms that scale well with data size is a current area of research.

Very recently, a class algorithms called \emph{piecewise deterministic Markov processes} (PDMP)\cite{fearnhead16} has been proposed with some surprising properties making them attractive to this task.
In particular, PDMPs -- like \emph{stochastic gradient descent} (SGD) \cite{robbins51} and \emph{stochastic variational inference} (SVI) \cite{hoffman13} -- can be used for simulation-based inference under an exchangeable model without performing any full-data computations.
Unlike SGD or SVI, however, PDMPs target the correct posterior distribution $\pi$ and do not entail the use of point estimates or distributional approximations.
Furthermore, subject to a one-off calculation, their computational cost can be $\c{O}(1)$ with respect to data size \cite{pollock16, bierkens16}.

These advantages have lead to increased interest in studying PDMPs -- particularly since, at present, only a small number of PDMPs invariant with respect to arbitrary target distributions are known.
These include the bouncy particle sampler \cite{peters12, bouchardcote15}, zig-zag \cite{bierkens16}, scalable Langevin exact \cite{pollock16}, and a few other variants whose comparative behavior is not yet well-understood.
In particular, all known PDMPs with constant deterministic dynamics also utilize linear transition functions.
In this work, we define a PDMP whose transition function is nonlinear and develop methods for computing with it.

Our contribution is purely theoretical. 
The process is derived in Section \ref{sec:pdmps}.
We present empirical evaluation in Sections \ref{sec:gauss} and \ref{sec:logit}. 
We discuss the implications of these results in Section \ref{sec:disc}.

\section{Piecewise Deterministic Markov Processes} \label{sec:pdmps}

Piecewise deterministic Markov processes are a class of stochastic processes first introduced by \textcite{davis84} and described in the Markov chain Monte Carlo (MCMC) context by \textcite{fearnhead16}.
All such algorithms evolve in part deterministically, and in part according to a Markov jump process.
These are fully described by three components.
\begin{enumerate}
\item[(1)] Deterministic Dynamics: a function $\Phi$ that determines the process' behavior between jumps, typically specified through a system of differential equations.
\item[(2)] Switching Rate: a function $\lambda$ that specifies the intensity of the jumps at each state.
\item[(3)] Transition Distribution: a probability measure $\mathbb{Q}$ that specifies what states the process jumps to.
\end{enumerate}

We refer the reader to \textcite{fearnhead16} for a detailed introduction.
One appealing property of PDMPs is that they can be simulated \emph{exactly} -- meaning with no discretization error -- through the use of techniques such as Poisson thinning.
This is because the switches evolve according to a nonhomogeneous Poisson process, which can be simulated by proposing switches from a Poisson process with greater intensity, and accepting or rejecting the proposals.
Other techniques, such as Poisson inversion, can also be considered \cite{fearnhead16}.
In between switches, the algorithm's behavior is deterministic, and can be calculated exactly provided $\Phi$ is tractable.

Another appealing property is that most PDMPs do not require the evaluation of the target distribution $\pi(\v{x})$ directly, and only depend on it through functions of $\nabla \ln \pi(\v{x})$. 
This term is amenable to unbiased estimation: $\nabla \ln \pi(\v{x})$ can be replaced by its expectation $\E[\grad\ln\pi(\v x)]$ with respect to some auxiliary variable, which can then be replaced with an unbiased estimate -- all without violating stationarity with respect to $\pi$.
Indeed \textcite{bouchardcote15}, \textcite{pollock16}, \textcite{bierkens16}, and \textcite{vanetti17} use PDMPs with stochastic gradients to obtain state-of-the-art performance on certain big data problems.
Moreover, \textcite{bierkens16} and \textcite{pollock16} have shown that, through introducing a control variate that can be computed using a one-off $\c{O}(N)$ calculation, the computational cost of certain such algorithms can be $\c{O}(1)$.
This makes PDMPs appealing in a big data setting.

\section{Piecewise Linear Markov Processes} \label{sec:plmps}

We now proceed to describe the class of PDMPs studied in this work, which includes the Bouncy Particle Sampler and Pure Reflection process as special cases.

\begin{definition} \label{def:plmp}
Consider a PDMP in the sense of \textcite{fearnhead16}.
Let $\pi(\v{v},\v{x}) = \pi(\v{v})\pi(\v{x})$ with $\pi(\v{v})$ standard multivariate Gaussian and $\pi(\v{x})$ the target distribution of interest.
Let $\mathbb{Q}$ to be the Dirac measure centered at $(\v{x}, F_{\v{x}}(\v{v}))$ for some function $F_{\v{x}}$, called the transition function.
Define the following.
\begin{enumerate}
\item Deterministic Dynamics: \vspace*{-4ex}

\<
\od{\v{x}}{t} &= v
&
\od{\v{v}}{t} &= 0
.
\>
\vspace*{-3ex}

\item Switching Rate:
\[
\lambda(\v{x}, \v{v}) = \max\{0,-\v{v} \cdot \nabla \ln \pi (\v{x}) ] \}
.
\]

\item Transition Function: \vspace*{-4.5ex}

\<
&F_{\v{x}}(\v{v})
&
&\text{s.t.} 
&
F^{-1}_{\v{x}}(\v{v}) &\text{ exists}
&
&\text{and}
&
&\pi(\v{x},\v{v}) \text{ is stationary.}
\>
\vspace*{-4.5ex}

\end{enumerate}
Call a PDMP satisfying these conditions a \emph{Piecewise Linear Markov Process}.
\end{definition}

Both the Pure Reflection process of \textcite{fearnhead16} and Bouncy Particle Sampler of \textcite{bouchardcote15} are examples within this subclass of PDMPs, with
\<
F_{\v{x}}(\v{v}) &= -\v{v}
&
&\text{and}
&
F_{\v{x}}(\v{v}) &= \v{v} - 2\frac{\v{v} \cdot \nabla \ln \pi (\v{x})}{||\nabla \ln \pi (\v{x})||^2}\nabla \ln \pi (\v{x})
\>
respectively.
In both cases, $F_{\v{x}}$ changes only the process' velocity $\v{v}$, and is linear -- the latter expression can be viewed as a reflection with respect to a hyperplane normal to $\nabla \ln \pi (\v{x})$.
In this work, we begin by asking the following question: does there exists a process within this class for which $F_{\v{x}}$ is nonlinear?

For such a process to exist, invariance must hold, which means that $\pi(\v{x}, \v{v})$ must be a zero of the Fokker-Planck Equation.
We proceed to derive this equation for processes given by Definition \ref{def:plmp}, which differ slightly from those considered by \textcite{fearnhead16}.

\begin{lemma} \label{lemma:fp}
The Fokker-Planck Equation for a Piecewise Linear Markov Process is given by
\[
\lambda(\v{x}, \v{v})\pi(\v{v}) - \lambda[\v{x}, F^{-1}_{\v{x}}(\v{v})] \pi[F^{-1}_{\v{x}}(\v{v}) \given \v{x}] \abs{\pd{F^{-1}_{\v{x}}(\v{v})}{\v{v}}} = -v \cdot [ \nabla \ln \pi (\v{x}) ] \, \pi(\v{v})
\]
where $\lambda(\v{x}, \v{v}) = \max\{0,-v \cdot \nabla \ln \pi (\v{x})\}$.
\end{lemma}

\begin{proof}
To simplify notation, we first consider general PDMPs -- within this lemma let $\v{z} = (\v{v},\v{x})$, $F(\v{z}) = (\v{x}, F_{\v{x}}(\v{v}))$, let $\Phi$ be the deterministic dynamics, and let $\mathbb{Q}$ be a Dirac measure centered at $F(\v{z})$.
We derive Fokker-Planck Equation from the infinitesimal generator
\[
\mathscr{A} f(\v{z}) = \Phi(\v{z}) \cdot \nabla f(\v{z}) + \lambda(\v{z}) \int_\Omega f(\v{z}') \dif \mathbb{Q}(\v{z}') - \lambda(\v{z}) f(\v{z})
\]
given by \textcite{davis84}, by finding its formal adjoint $\mathscr{A}^*$ satisfying
\[
\int_\Omega \pi(\v{z}) \mathscr{A} f(\v{z}) \dif\v{z} = \int_\Omega f(\v{z}) \mathscr{A}^* \pi(\v{z}) \dif\v{z}
.
\]
We proceed as follows. 
First, note that by linearity we may find the formal adjoint of $\mathscr{A}$ component-wise.
It is shown in \textcite{fearnhead16} that the components
\<
&\hspace*{2ex}\Phi(\v{z}) \cdot \nabla f(\v{z})\hspace*{2ex}
&
&\text{and}
&
& - \lambda(\v{z}) f(\v{z})
\>
map to
\<
&-\sum_{i=1}^{2p} \pd{\Phi_i(\v{z})}{z_i} \pi(\v{z})
&
\text{and}&
&
& - \lambda(\v{z}) \pi(\v{z})
\>
respectively. For the remaining component, we can write
\<
\int_\Omega \pi(\v{z}) \lambda(\v{z}) \int_\Omega f(\v{z}') \dif \mathbb{Q}(\v{z}') \dif\v{z} &= \int_\Omega \pi(\v{z}) \lambda(\v{z}) f[F(\v{z})] \dif\v{z}
\nonumber
\\
&= \int_\Omega f(\v{\tilde{z}}) \pi[F^{-1}(\v{\tilde{z}})] \lambda[F^{-1}(\v{\tilde{z}})]  \abs{\frac{\partial F^{-1}(\v{\tilde{z}})}{\partial\v{\tilde{z}}}} \dif\v{\tilde{z}}
\>
where the change of variables is justified because $F_{\v{x}}$ is assumed invertible everywhere, except possibly a set of measure zero, in which case we may divide $\Omega$ accordingly and invert $F_{\v{x}}$ piecewise.
Note the presence of a Jacobian term that does not explicitly appear in the derivation of \textcite{fearnhead16} because they consider a slightly different case.
Thus we have
\[
\mathscr{A}^* \pi(\v{z}) = -\sum_{i=1}^{2p}  \pd{\Phi_i(\v{z})}{z_i} \pi(\v{z}) + \pi[F^{-1}(\v{z})] \lambda[F^{-1}(\v{z})]  \abs{\frac{\partial F^{-1}(\v{z})}{\partial\v{\tilde{z}}}} - \lambda(\v{z}) \pi(\v{z})
.
\]
For $\pi(\v{z})$ to be invariant, we must have $\mathscr{A}^* \pi(\v{z}) = 0$.
Consider now our case, where $F(\v{x},\v{v}) = (\v{x}, F_{\v{x}}(\v{v}))$, and $\pi(\v{z}) = \pi(\v{v})\pi(\v{x})$.
The expression then simplifies to the desired result.
\end{proof}

Both the Pure Reflection Process of \textcite{fearnhead16} and the Bouncy Particle Sampler of \textcite{bouchardcote15} are processes within Definition \ref{def:plmp}.
Observe that they both solve the Fokker-Planck Equation by letting
\<
||F^{-1}_{\v{x}}(\v{v})|| &= ||\v{v}||
&
-F^{-1}_{\v{x}}(\v{v}) \cdot -\nabla\ln\pi(\v{x}) &= \v{v} \cdot -\nabla\ln\pi(\v{x})
&
\abs{\pd{F^{-1}_{\v{x}}(\v{v})}{\v{v}}} &= 1
.
\>
Both of these solutions are magnitude-preserving.
This motivates us to further ask: are there solutions that are not magnitude-preserving?
We now proceed to find such a solution.

\begin{thm} \label{prop:solution}
Let $r = ||\v{v}||$ and $\theta$ be the angle between $\nabla \ln \pi(\v{x})$ and $\v{v}$ along the hyperplane spanned by both vectors.
Consider a transition function $F_{\v{x}}$ which maps $\v{v}$ to another vector on that hyperplane, which is fully determined by the coordinates $r', \theta'$.
Suppose that $r'$ is only a function of $\theta$ and $\theta'$ is only a function of $r$.
Then, letting $k$ be a positive constant and $p$ be the dimension of $\v{v}$, we have that for every $r$, if we take $\theta'$ to be the solution of the differential equation
\[
\frac{\dif \theta'(r)}{\dif r} = \frac{k r^p \exp\cbr{\frac{r^2}{-2}}}{\cos\sbr{\theta'(r)} \sin^{p-2}\sbr{\theta'(r)}} 
\]
subject to the boundary conditions $\theta'(0) = 0$ and $\theta'(\infty) = \pi/2$ which fully determine $k$, and if we take $r'(\theta)$ to be the above solution's inverse, the resulting PDMP is $\pi$-invariant.
\end{thm}

\begin{proof}
By Lemma \ref{lemma:fp}, the Fokker-Planck Equation is
\[
\max\{0, \v{v} \cdot \nabla \ln \pi(\v{x})\}\pi(\v{v}) - \max\{0, F^{-1}_{\v{x}}(\v{v}) \cdot \nabla \ln \pi(\v{x})\} \pi[F^{-1}_{\v{x}}(\v{v})] \abs{\pd{F^{-1}_{\v{x}}(\v{v})}{\v{v}}} = \v{v} \cdot \nabla \ln \pi(\v{x}) \ \pi(\v{v})
\]
which for $\v{v} \cdot \nabla \ln \pi(\v{x}) > 0$ is always true provided $F_{\v{x}}(\v{v})^{-1} \cdot \nabla \ln \pi(\v{x}) < 0$, which we henceforth assume.
Consider $\v{v} \cdot \nabla \ln \pi(\v{x}) < 0$, and suppose $F_{\v{x}}(\v{v})^{-1} \cdot \nabla \ln \pi(\v{x}) > 0$.
Substituting in $\pi(\v{v}) \propto \exp\cbr{\frac{||\v{v}||^2}{-2}}$, we can write
\[
- F^{-1}_{\v{x}}(\v{v}) \cdot \nabla \ln \pi(\v{x}) \, \exp\cbr{\frac{||F^{-1}_{\v{x}}(\v{v})||^2}{-2}} \abs{\pd{F^{-1}_{\v{x}}(\v{v})}{\v{v}}} = \v{v} \cdot \nabla \ln \pi(\v{x}) \, \exp\cbr{\frac{||\v{v}||^2}{-2}}
.
\]
Now, transform to hyperspherical coordinates coordinates, by letting 
\<
v_1 = r &\cos(\theta) 
&
v_2 = r \sin(&\theta) \cos(\phi_1)
\nonumber
\\
&\vdots
&
&\vdots
\nonumber
\\
v_{p-1} = r \sin(&\theta) \prod_{i=1}^{p-2} \sin(\phi_i) \cos(\phi_{p-1})
&
v_p = r \sin(&\theta) \prod_{i=1}^{p-2} \sin(\phi_i)
\>
where $\theta$ is the angle on the hyperplane spanned by $\v{v}$ and $\nabla \ln \pi(\v{x})$, and $\phi_i$ are angles on an arbitrary set of hyperplanes orthogonal to $\v{v}$ and $\nabla \ln \pi(\v{x})$. 
For this transformation, we have the identities
\<
||\v{v}|| &= r
&
\frac{\v{v} \cdot \nabla \ln \pi(\v{x})}{||\v{v}||\,||\nabla \ln \pi(\v{x})||} = \cos(\theta)
\>
and letting $r', \theta', \v{\phi}'$ be the coordinates under $F_{\v{x}}^{-1}$, i.e. functions of $r, \theta, \v{\phi}$, the Fokker-Planck equation becomes
\[
- \cos(\theta') r' \exp\cbr{\frac{{r'}^2}{-2}} \abs{\pd{F^{-1}_{\v{x}}(\v{v})}{\v{v}}} = \cos(\theta) r \exp\cbr{\frac{r^2}{-2}}
.
\]
We can decompose the Jacobian into
\[
\abs{\pd{F^{-1}_{\v{x}}(\v{v})}{\v{v}}} = \abs{\pd{F^{-1}_{\v{x}}(\v{v})}{(r',\theta',\v{\phi}')}} \abs{\pd{(r',\theta',\v{\phi}')}{(r,\theta,\v{\phi})}} \abs{\pd{(r,\theta,\v{\phi})}{(\v{v})}}
\]
which, since the Jacobian for hyperspherical coordinates is 
\[
\abs{\frac{\partial(r,\theta,\v{\phi})}{\partial(\v{v})}} = \sbr{r^{p-1} \sin^{p-2}(\theta) \prod_{i=1}^{p-2} \sin^{p-1-i}(\phi_i)}^{-1}
\]
yields the system
\<
- \cos(\theta') &\sin^{p-2}(\theta') \sbr{\prod_{i=1}^{p-2} \sin^{p-1-i}(\phi_i')} {r'}^p \exp\cbr{\frac{r'(r, \theta)^2}{-2}} \abs{\frac{\partial(r',\theta', \v{\phi}')}{\partial(r,\theta,\v{\phi})}} =
\nonumber
\\
&= \cos(\theta) \sin^{p-2}(\theta) \sbr{\prod_{i=1}^{p-2} \sin^{p-1-i}(\phi_i)}  r^p \exp\cbr{\frac{r^2}{-2}}
.
\>
Now, suppose that we are interested in solutions where $\v{\phi}' = \v{\phi}$, $\theta'$ is only a function of $r$ and $r'$ is only a function of $\theta$.
The Jacobian is just
\[
\abs{\frac{\partial(r',\theta', \v{\phi}')}{\partial(r,\theta, \v{\phi})}}
=
\begin{vmatrix}
\m{I}_{p-1} & \v{0} & \v{0}
\\[1ex]
\v{0} & \frac{\partial r'}{\partial r} & \frac{\partial r'}{\partial \theta}
\\[1ex]
\v{0} & \frac{\partial \theta'}{\partial r} & \frac{\partial \theta'}{\partial \theta}
\end{vmatrix}
=
\begin{vmatrix}
0 & \frac{\partial r'}{\partial \theta}
\\[1ex]
\frac{\partial \theta'}{\partial r} & 0
\end{vmatrix}
=
\abs{ \frac{\partial r'}{\partial \theta} \frac{\partial \theta'}{\partial r} }
.
\]
Under these assumptions, the Fokker-Planck Equation becomes
\[
-\cos\sbr{\theta'(r)} \sin^{p-2}\sbr{\theta'(r)} \, r'(\theta)^p \exp\cbr{\frac{r'(\theta)^2}{-2}} \abs{ \frac{\partial r'}{\partial \theta} \frac{\partial \theta'}{\partial r} } = \cos(\theta) \sin^{p-2}(\theta) \, r^p \exp\cbr{\frac{r^2}{-2}}
\]
which we can multiply on both sides by an arbitrary constant $k$, then factorize into the system
\<
-\cos\sbr{\theta'(r)} \sin^{p-2}\sbr{\theta'(r)} \abs{ \frac{\dif \theta'(r)}{\dif r} } &= k \, r^p \exp\cbr{\frac{r^2}{-2}}
\nonumber
\\
k\,r'(\theta)^p \exp\cbr{\frac{r'(\theta)^2}{-2}} \abs{ \frac{\dif r'(\theta)}{\dif \theta} } &= \cos(\theta) \sin^{p-2}(\theta) 
\>
and rewrite as
\<
\frac{\dif \theta'(r)}{\dif r} &= \frac{k r^p \exp\cbr{\frac{r^2}{-2}}}{\cos\sbr{\theta'(r)} \sin^{p-2}\sbr{\theta'(r)}} 
&
\frac{\dif r'(\theta)}{\dif \theta} &= \frac{\cos(\theta) \sin^{p-2}(\theta)}{k\,r'(\theta)^p \exp\cbr{\frac{r'(\theta)^2}{-2}}}
.
\>
Notice that these differential equations are reciprocals of one another: therefore, subject to identical initial conditions, $\theta'(r)$ and $r'(\theta)$ are inverse functions.
We have thus shown that $F_{\v{x}} = F_{\v{x}}^{-1}$, and therefore need not consider the inversion.
We impose the boundary conditions
\<
\theta'(0) &= 0
&
\theta'(\infty) &= \pi/2
\>
under which the above differential equations can be solved analytically.
Since these solutions are monotonic, the result follows.
\end{proof}

Though the above differential equations can be solved analytically, computation using them is intractable because they are not numerically stable due to the presence of large powers.
Indeed, for moderate $p$, to satisfy boundary conditions the constant $k$ needs to be taken closer to zero than the smallest positive number available in double precision arithmetic.
As a result, we cannot proceed directly.
We instead consider the differential equation's asymptotic form for large $p$ -- this introduces some approximation error that vanishes in high dimension.

\begin{proposition} \label{cor:asym}
For large $p$, we have
\[
\theta'(r) \approx \frac{\pi}{2} - \frac{1}{2} \sqrt{\frac{-8}{p-2} \ln\Phi\sbr{\sqrt{2}\del{r - \sqrt{p}}} }
\]
where $\Phi$ is the CDF of a unit Gaussian, in the sense that $\theta'(r)$ is the solution of a differential equation whose right-hand side is the pointwise limit as $p \goesto \infty$ of the equation in Proposition \ref{prop:solution}.
\end{proposition}

\begin{proof}
It is a standard result that
\[
\lim_{p\goesto\infty} \abs{ \frac{2^{1 - (p+1)/2}\sqrt{\pi}}{\Gamma\sbr{(p+1)/2}} r^p \exp\cbr{\frac{r^2}{-2}} - \exp\cbr{-(r - \sqrt{p})^2} } = 0
\]
for all $r \in \R^+$, as the former is the density of a $\raisebox{0.4ex}{$\chi$}$ distribution, and that
\[
\lim_{p\goesto\infty} \abs{ \cos(\theta) \sin^p(\theta) - \del{\theta - \frac{\pi}{2}} \exp\cbr{\frac{(\theta - \pi/2)^2}{-2/p}} } = 0
\]
for $\theta \in [0,\pi/2]$.
Therefore, the limiting form for our differential equation is
\[
\frac{\dif \theta'(r)}{\dif r} = \frac{k' \exp\cbr{-(r - \sqrt{p})^2}}{\del{\theta'(r) - \frac{\pi}{2}} \exp\cbr{\frac{(\theta'(r) - \pi/2)^2}{-2/(p-2)}}}
\]
for some constant $k'$, which has analytic solution
\[
\theta'(r) = \frac{\pi}{2} - \frac{1}{2} \sqrt{c_1 - \frac{8}{p-2} \ln\sbr{\pm 1 + c_2 \f{erf}\del{r - \sqrt{p}}} }
\]
for arbitrary constants $c_1, c_2$.
We must choose $c_1, c_2$ such that $F_{\v{x}}$ is invertible, and $r$ is positive.
If we set $\theta'(0) = 0$, $\theta'(\infty) = \pi/2$, we obtain that $\pm$ should be taken to be $+$ and
\<
c_2 &= \frac{\exp\cbr{\frac{\pi^2 (p-2)}{8}} - 1}{\exp\cbr{\frac{\pi^2 (p-2)}{8}} + 1} \approx 1
&
c_1 &= \frac{8\ln(1 + c_2)}{\pi^2 (p-2)} \approx \frac{8\ln(2)}{\pi^2 (p-2)}
\>
which yields the desired result.
\end{proof}

For $c_1$ and $c_2$ as above, the solution is strictly increasing and positive everywhere, except possibly on an interval near the origin.
From a practical perspective, this is not a concern, as the probability of landing in those states is exceedingly small and was never occurred in our simulations.
The inverse function $r'(\theta)$ is obtained numerically, which can easily be done as $\theta'(r)$ is one-dimensional.
This completes our derivation. 

\section{Example: Independent Gaussian Target} \label{sec:gauss}

To understand the algorithm's behavior, we implemented it for a standard multivariate Gaussian target and compared it against the bouncy particle sampler.
We examined three targets with dimension $p=10, 100, 1,\!000$.
All were implemented using Poisson thinning with identical velocity-dependent switching rate bound $5\,||\v{v}||$ which was never exceeded outside of burn-in.
Velocity was resampled according to a homogeneous Poisson process with intensity $0.2$.
Each algorithm was given a fixed computational budget consisting of $100,\!000$ gradient evaluations, and started from initial values of $(10,..,10)$ selected to be away from the target mode.
This implementation avoids using analytic properties of Gaussians to better mimic real-world scenarios.

Trace plots of the resulting chains can be seen in Figure \ref{fig:gauss}.
It can be seen that for $p=10$, both algorithms produce similar output.
For $p=100$, we find that the algorithm approximately converged to the correct mean and variance slightly faster than the bouncy particle sampler.
Neither algorithm performed effectively for $p=1,\!000$.

\begin{figure*}
\includegraphics{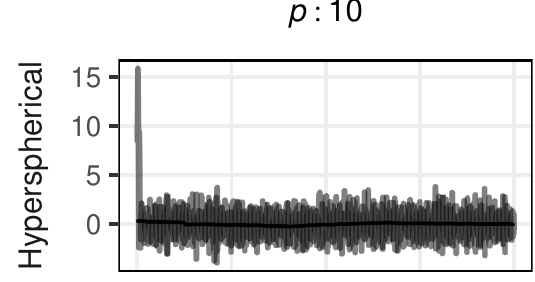}
\includegraphics{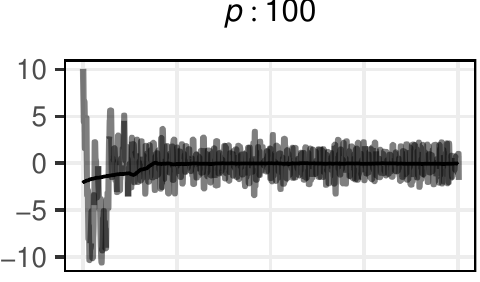}
\includegraphics{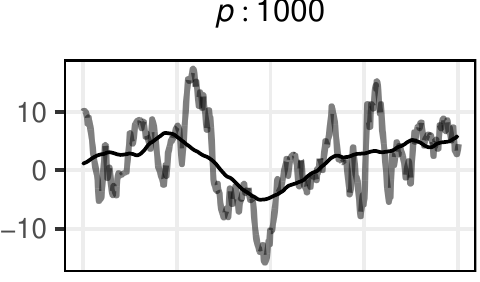}

\includegraphics{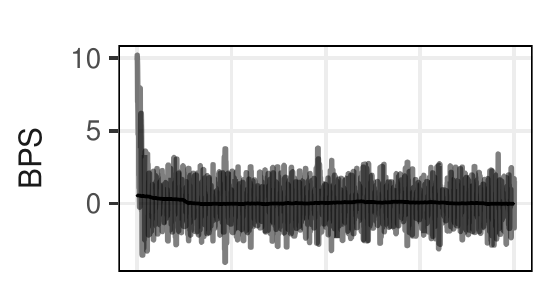}
\includegraphics{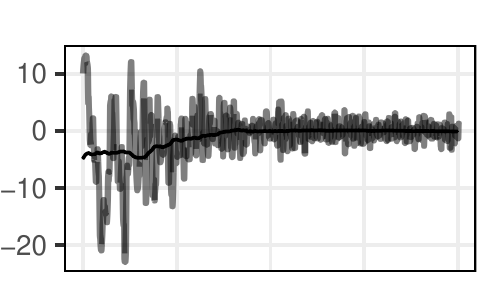}
\includegraphics{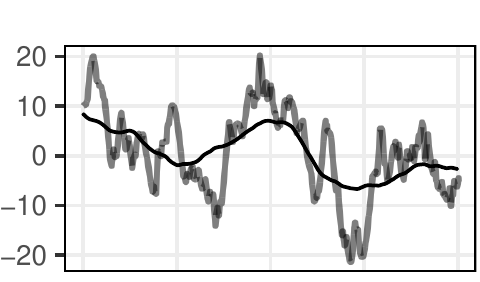}
\caption{Trace plots for the first coordinate of a multivariate Gaussian of dimension $p \in \{10, 100, 1000\}$. \emph{Hyperspherical} refers to the process derived in Section \ref{sec:plmps}, \emph{BPS} refers to the bouncy particle sampler.}
\label{fig:gauss}
\end{figure*}

\section{Example: Bayesian Logistic Regression} \label{sec:logit}

To examine the performance on a Bayesian model with known correct answer, we implemented the algorithm for a Bayesian Logistic Regression problem with synthetic data. 
Data was generated by taking
\<
\v{x}_i &\overset{\f{iid}}{\sim}\f{N}_p(\v{0},\m{I})
&
\v{\beta} &= (1.3,4,-1,1.6,5,-2, \v{0}_{p-6})^T
&
\v{y} &\sim\f{Ber}\sbr[1]{\Psi(\m{X} \v{\beta})} \,.
\>
where $\Psi$ is the logistic function.
We used the logistic regression model 
\<
\v{y} \given \v\beta &\sim\f{Ber}\sbr[1]{\Psi(\m{X} \v\beta)}
&
\v\beta &\sim\f{N}_p (\v{0},10^{-3} \, \m{I})
.
\>
We selected $N=1,\!000,\!000$ and $p=100$, and implemented the bouncy particle sampler as well as the algorithm of Section \ref{sec:plmps}.
Both utilized Poisson thinning with a constant switching rate bound $\hat\lambda = 5000\,||\v{v}||$, which was selected to be sufficiently large to ensure it was not exceeded more than $1\%$ of the time.
Velocity was resampled according to a homogeneous Poisson process with intensity $10$.

Computation was performed as follows.
First, a point estimate $\v{\hat\beta}$ of the posterior mode was obtained using stochastic gradient descent, consisting of $100,\!000$ steps, each with a batch size of $10$, using a total of $N$ data points.
Then, the data was used to precompute $\grad\ln\pi(\v{\hat\beta})$, which was then used to implement the control variate of \textcite{bierkens16} and \textcite{pollock16}.
Finally, sampling was performed, using $100,\!000$ stochastic gradient evaluations, each with a batch size of $10$, with the control variate used to reduce variance.

Results can be seen in Figure \ref{fig:logit}.
Given the extremely limited nature of our computational budget -- a total of $3N$ evaluations of $(\v{x}_i, y_i)$ pairs -- both algorithms obtained reasonable posterior samples, performing similarly.
We find this remarkable: standard MCMC methods such as Gibbs sampling \cite{casella92} and Hamiltonian Monte Carlo \cite{betancourt17a} would not generally produce useful output under such constraints.

\begin{figure*}
\includegraphics{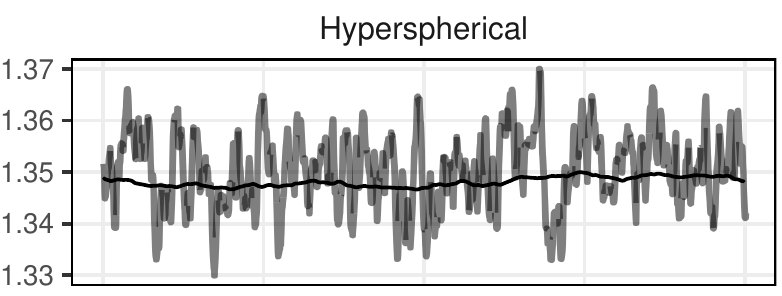}
\includegraphics{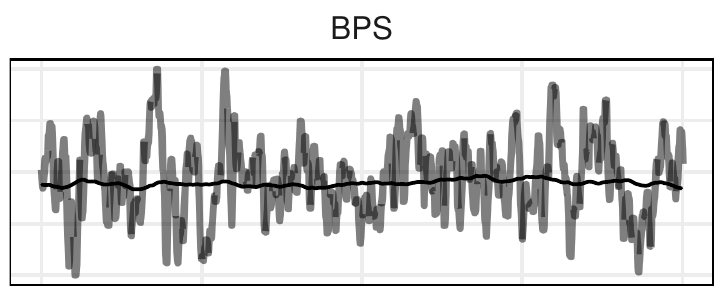}
\caption{Trace plots for the first coordinate of the logistic regression target distribution. \emph{Hyperspherical} refers to the process derived in Section \ref{sec:plmps}, \emph{BPS} refers to the bouncy particle sampler.}
\label{fig:logit}
\end{figure*}

\section{Discussion} \label{sec:disc}

The PDMP constructed in Section \ref{sec:plmps} performs slightly better than the bouncy particle sampler for Gaussian targets of moderate dimension. 
This is because its transition function is nonlinear and non-magnitude-preserving -- this helps the process avoid getting stuck in high-dimensional orbits by making it easier to move perpendicular to the contours of the target distribution.
Unfortunately, the overall improvement is rather limited -- non-magnitude-preserving transitions appear to us to be necessary but not sufficient for efficiency in high dimension.

Our results in Section \ref{sec:logit} replicate the behavior of other PDMPs on big data problems explored in detail by \textcite{bouchardcote15}, \textcite{bierkens16}, \textcite{pollock16}, and \textcite{vanetti17}.
It is clear that these algorithms can achieve state-of-the-art performance in this setting through the use of subsampling and precomputed control variates.
For logistic regression, this technique is attractive because the posterior mode is easily obtained using classical techniques.

One difficulty with PDMPs well-illustrated by our work can be seen in the trace plots under the Gaussian target with $p=1,\!000$.
The trajectories produced by both the process of Section \ref{sec:plmps} and the bouncy particle sampler, while clearly not indicative of good mixing, are also not entirely atypical to those often seen in practice.
In standard MCMC settings, such trace plots indicate diffusive behavior, which may lead practitioners to conclude that since the Markov chain is moving slowly through the state space, variance is likely to be underestimated.
For PDMPs, this doesn't follow: intuitively, it is possible for a non-reversible algorithm to \emph{always} move rapidly through the state space, and yet still converge slowly due to moving primarily in directions orthogonal to those needed to ensure good mixing.
Our use of hyperspherical coordinates makes the above easy to visualize: a non-reversible process can move rapidly in the $\theta$ and $\v{\phi}$ dimensions while moving arbitrarily slowly in the $r$ dimension.
Thus, non-reversible MCMC methods require additional care to diagnose convergence and ensure posterior estimates are reliable.

Further research in PDMPs is needed to understand their behavior on high-dimensional targets.
Our use of hyperspherical coordinates to derive a PDMP with a nonlinear transition function may yield improvement for certain targets of moderate dimension.
Many PDMPs resemble Hamiltonian Monte Carlo \cite{betancourt17a}, so it may be possible to connect current work with existing theory in that area.
We hope that with additional ideas substantially larger improvements are possible.

\section*{Acknowledgments}

We are grateful to Georgi Dinolov, David Draper, Mark Girolami, David Parks, Daniele Venturi, and Yuanran Zhu for their thoughts.
Membership on this list does not imply agreement with the ideas presented nor responsibility for errors that may inadvertently be present.

\section*{References}

\printbibliography

\end{document}